\documentclass[12pt]{article}%

\usepackage[cm]{fullpage}%
\usepackage{graphicx}%
\usepackage{amsmath}%
\usepackage{xspace}%
\usepackage{xcolor}%
\usepackage{paralist}%
\usepackage{mleftright}%

\usepackage[thmmarks]{ntheorem}%
\theoremseparator{.}%

\usepackage{hyperref}%
\hypersetup{%
   breaklinks,%
   colorlinks=true,%
   linkcolor=[rgb]{0.45,0.0,0.0},%
   citecolor=[rgb]{0,0,0.45}
}

\numberwithin{figure}{section}%
\numberwithin{table}{section}%
\numberwithin{equation}{section}%

\usepackage{color}

\newcommand{\xbeginlgox}{\begin{minipage}{1in}\begin{tabbing}
           \quad\=\qquad\=\qquad\=\qquad\=\qquad\=\qquad\=\qquad\=\kill}
        \newcommand{\xendlgox}{\end{tabbing}\end{minipage}}

\newenvironment{program}{
   \begin{minipage}{4.0in}
   \begin{tabbing}
       \ \ \ \ \= \ \ \ \= \ \ \ \ \= \ \ \ \ \= \ \ \ \ \=
      \ \ \ \ \= \ \ \ \ \= \ \ \ \ \= \ \ \ \ \=
      \ \ \ \ \= \ \ \ \ \= \ \ \ \ \= \ \ \ \ \= \kill
}{
   \end{tabbing}
   \end{minipage}
}

\DefineNamedColor{named}{RedViolet} {cmyk}{0.07,0.90,0,0.34}

\newcommand{\nnY}[2]{\nu^{}_{#2}(#1)}
\newcommand{\NNY}[2]{N_{\leq #2}(#1)}

\newcommand{\myqedsymbol}{\rule{2mm}{2mm}}

\newcommand{\cardin}[1]{\left| {#1} \right|}%
\newcommand{\Graph}{{G}}%
\newcommand{\Vertices}{{V}}%
\newcommand{\Edges}{{E}}%

\newcommand{\dirEdge}[2]{{{#1} \rightarrow {#2}}}

\newcommand{\pth}[2][\!]{\mleft({#2}\mright)}%

\newcommand{\brc}[1]{\left\{ {#1} \right\}}

\newcommand{\distG}[2]{d_{\Graph}\pth{#1, #2}}
\renewcommand{\th}{th\xspace} \newcommand{\atgen}{\symbol{'100}}
\newcommand{\SarielThanks}[1]{\thanks{Department of Computer Science;
      University of Illinois; 201 N. Goodwin Avenue; Urbana, IL,
      61801, USA; {\tt sariel\atgen{}illinois.edu}; {\tt
         \url{http://sarielhp.org/}.} #1}}

\newtheorem{theorem}{Theorem}[section]
\newtheorem{lemma}[theorem]{Lemma}

\theoremstyle{plain}%
\theoremheaderfont{\sf}%
\theorembodyfont{\upshape}%
\newtheorem{remark}[theorem]{Remark}%

\newcommand{\HLink}[2]{\hyperref[#2]{#1~\ref*{#2}}}

\newcommand{\thmlab}[1]{{\label{theo:#1}}}
\newcommand{\thmref}[1]{\HLink{Theorem}{theo:#1}}

\newcommand{\pbrcx}[1]{\left[ {#1} \right]}%
\newcommand{\Prob}[1]{\mathop{\mathbf{Pr}}\!\pbrcx{#1}}

\definecolor{blue25}{rgb}{0, 0, 11}
\newcommand{\emphic}[2]{%
   \textcolor{blue25}{%
      \textbf{\emph{#1}}}%
   \index{#2}}

\newcommand{\emphi}[1]{\emphic{#1}{#1}}

\theoremheaderfont{\em}%
\theorembodyfont{\upshape}%
\theoremstyle{nonumberplain}
\theoremseparator{}
\theoremsymbol{\myqedsymbol}
\newtheorem{proof}{Proof:}

\begin{document}

\title{Computing the $k$ Nearest-Neighbors for all Vertices via
   Dijkstra}

\author{Sariel Har-Peled\SarielThanks{Work on this paper
      was partially supported by a NSF AF awards
      CCF-1421231, and 
      CCF-1217462.  
   }}

\date{\today}

\maketitle

\begin{abstract}
    We are given a directed graph $\Graph = (\Vertices,\Edges)$ with
    $n$ vertices and $m$ edges, with positive weights on the edges,
    and a parameter $k >0$. We show how to compute, for every vertex
    $v \in \Vertices$, its $k$ nearest-neighbors. The algorithm runs
    in $O( k ( n \log n + m ) )$ time, and follows by a somewhat
    careful modification of Dijkstra's shortest path algorithm.

    This result is probably folklore, but we were unable to find a
    reference to it -- thus, this note.
\end{abstract}


\section{The problem}

Let $\Graph = (\Vertices, \Edges)$ be an undirected graph with $n$
vertices, $m$ edges, and with positive weights on the edges. For
$u,v \in \Vertices$, let $\distG{u}{v}$ denote the shortest path
distances in $\Graph$, and assume the sake of the simplicity of
exposition that all the non-trivial shortest path distances in the
graph are distinct.

A vertex $u = \nnY{v}{i}$ is the \emphi{$i$\th nearest-neighbor} to
$v$, if one partition $\Vertices$ into three disjoint sets
$C, \{u_i\}, F$, such that
\begin{inparaenum}[(i)]
    \item $\cardin{C} = i-1$, and
    \item $\forall c \in C$, $\forall f \in F$ we have
    $\distG{c}{v} < \distG{u}{v} < \distG{f}{v}$.
\end{inparaenum}
Observe that $\nnY{v}{0} = v$.  For any integer $k$, the \emphi{$k$
   nearest-neighbors} to $v$, are the members of the set
\begin{align*}
    \NNY{v}{k} = \brc{ \Bigl. \nnY{v}{0}, \nnY{v}{1}, \ldots,
   \nnY{v}{k-1}};
\end{align*}
that is, they are the $k$ vertices in $\Graph$ that are closest to
$v$.

Our purpose here is to compute for every vertex $v \in \Vertices$, the
set $\NNY{v}{k}$; that is, to compute for $v$ the $k$ distinct
vertices closest to it.

\section{A simple (but slower) randomized algorithm}

\subsection{Algorithm}
Let $t = O( k \log n)$, and let $R_1, \ldots, R_t$ be random samples
from $\Vertices$, where every vertex is picked into the $i$\th random
sample with probability $p=1/k$. Now, compute for every vertex
$v \in \Vertices$ its distance from its nearest-neighbor in $R_i$, for
all $i$. For a specific $i$, this can be done by performing Dijkstra
in $\Graph$ starting from all the vertices of $R_i$ simultaneously
(i.e., we create a fake source vertex $s$, add it to $\Graph$, and add
edges of weight $0$ from $s$ to all the vertices of $R_i$, and perform
regular Dijkstra from $s$). For a vertex $v$, let $D'(v)$ be its $t$
candidate nearest-neighbors computed by these $t$ executions of
Dijkstra. Using hashing, remove duplicates in $D'(v)$ (i.e., a vertex
$s$ might be the nearest-neighbor for $v$ in several of these
executions -- note however that in such a case it is always the same
distance). Now, compute the $k$ vertices with the smallest numbers
associated with them in $D'(v)$, and let $D(v)$ be the resulting set
of nearest-neighbors. We claim that, for all $v$, the sets $D(v)$ are
the desired $k$ nearest-neighbors.

\newcommand{\Event}{\mathcal{E}}

\subsection{Analysis}

For the running time, observe that performing  Dijkstra shortest
path algorithm $t$ times takes $O( t (n + m \log n) )$ time. All the
other work is dominated by this. 

As for correctness, we need to argue that if $u$ is the $j$\th
nearest-neighbor to $v$, for $j \leq k$, then $u \in D(v)$. To this
end, consider the event $\Event_i$ that $u \in R_i$, and none
of the vertices of $\NNY{v}{j-1}$ are in $R_i$. We have that 
\begin{align*}
    \alpha%
    =%
    \Prob{\Event_i}%
    =%
    p (1-p)^{j-1}%
    \geq%
    (1-1/k)^{k} /k%
    \geq%
    \exp( -1/2k)^k /k%
    \geq%
    1/10k.
\end{align*}
If this happens, then $u \in D'(v)$, which in turn implies that
$u \in D(v)$, as desired. The probability that none of the events
$\Event_1, \ldots, \Event_t$ happens is
$(1-\alpha)^t \leq (1-1/10k)^t < 1/n^{c}$ by making $t$ sufficiently
large, where $c$ is an arbitrary constant. Since there are $n k$ pairs
of $(u,v)$ such that $u$ is one of the $k$ nearest neighbors to $v$,
it follows that the probability this algorithm fails is at most
$1/n^{c-2}$.

\subsection{The result}

\begin{lemma}%
    %
    Given a directed graph $\Graph = (\Vertices,\Edges)$ with $n$
    vertices and $m$ edges, with positive weights on the edges, and a
    parameter $k >0$, one can compute, in
    $O\pth{ \pth{ n \log n + m} k \log n}$ time, for every vertex
    $v \in \Vertices$, its $k$ nearest-neighbors in $\Graph$. The
    algorithm succeeds with high probability.
\end{lemma}

The above randomized algorithm is inspired by the Clarkson-Shor
technique \cite{cs-arscg-89} and this trick is useful in many other
scenarios.

\section{A faster algorithm}

\subsection{The algorithm -- a first try}

We are going to run (conceptually) $n$ copies of the shortest-path
algorithm simultaneously.  In particular, let $A_v$ be the
shortest-path algorithm starting from the vertex $v$, for all
$v \in \Vertices$.  We use global heap for the events for all these
algorithms together. Here, every event of $A_v$ would be indexed by
the source vertex $v$ associated with this algorithm. The algorithm is
going to maintain for each vertex the set of $D'(v)$ of
nearest-neighbors found so far, and a count $c_v = \cardin{D'(v)}$.

Now, when the algorithm extract the next vertex to be visited (i.e.,
the one of the with lowest candidate distance), we get a triple
$(v, s, d)$, where $v$ is the vertex to be handled, $s$ is the source
vertex, and $d$ is the proposed distance. Using a hash table, we check
in constant time whether $v$ has $s$ as one of its computed nearest
neighbors (i.e., check if $s \in D'(v)$), and if so the algorithm
continues to the next iteration. Otherwise, the algorithm
\begin{compactenum}[\qquad(i)]
    \item adds $s$ to $D'(v)$,
    \item increase $c_v$, and 
    \item perform the standard relax operation from $v$ for all the
    outgoing edges of $v$ (these operations are ``marked'' by
    the source vertex $s$ they are being done for).

    Specifically, consider an edge $\dirEdge{v}{z}$ being relaxed,
    during the handling of the event $(v,s,d)$. The new associated
    event is $\pth{z, s, d+ w(\dirEdge{v}{z})}$, where
    $w(\dirEdge{v}{z})$ is the weight of the edge $\dirEdge{v}{z}$. If
    $s \notin D'(z)$ the algorithm ``schedule'' this event by
    inserting it to the heap, otherwise it ignores it.
\end{compactenum}

\medskip

The basic observation is that once a vertex was visited by $k$ of
these parallel executions, it is no longer needed, and it can be
``disabled'' blocking it from being visited by any other Dijkstra
(i.e., as soon as $c_v = k$). From this point on, the algorithm
ignores update operations for triplets of the form $(v, \cdot,
\cdot)$.

\subsubsection{Analysis}

\begin{lemma}
    The algorithm computes correctly, for each vertex $v \in
    \Vertices$, its $k$ nearest-neighbors in $\Graph$.
\end{lemma}
\begin{proof}
    By induction on the distance being computed. For $\ell=0$, for
    each vertex $v \in \Vertices$, the vertex $v$ is its own $0$\th
    nearest-neighbor, of distance $0$, and it was computed
    correctly. 

    Assume all the relevant distances $< \ell$ were computed correctly
    (they are $0,1,\ldots, k$ nearest-neighbor distances for some
    pairs of vertices in the graph).

    So consider a vertex $v$, where its $i$\th nearest-neighbor is
    $s$, let its (real) shortest path from $s$ to $v$ be $\pi = s =
    v_0 v_1 v_2 \ldots v_t = v$, and assume the length of $\pi$ is
    $\ell$.

    The key observation is that for all the vertices on $\pi$, the
    vertex $s$ must be one of their first $i$
    nearest-neighbors. Indeed, if not, then there exists $i$ vertices
    in $\Graph$ that are closer to $v$ than $s$, which is a
    contradiction.

    Thus, by induction all the relevant distances for $v_0, v_1,
    \ldots, v_{t-1}$ were computed correctly by the algorithm, and as
    such, this path would be considered by the algorithm. Namely, the
    $i$\th distance to $v$ would be set correctly.
\end{proof}

\paragraph{Running time.}
Since for a vertex $v$, the counter $c_v$ can be increased only $k$
times, it follows that every edge participates in $k$ relax
operations. As such, the total number of relax operations handled by
this algorithm is $O(k m)$, and this bounds the maximum size of the
global heap. The global heap might perform $O(k m)$ extract-min
operations on the global heap. As such, the overall running time of
this algorithm is $O( k m \log n)$.

\subsection{Speeding up the algorithm}

Now, our purpose is to improve the running time -- this requires some
cleverness with the data-structures being used. 

Our first task is to avoid having extract-min operations in the global
heap involving a vertex $v$ such that $c_v = k$. To this end, for each
vertex $v \in \Vertices$, the algorithm maintains a separate queue
$Q_u$ that handles all the events for $u$. Every source vertex
$s \in \Vertices$ might maintain at most one value at $Q_u$. Every
such queue would maintain its current minimum, and would update it in
the global queue $Q$ as necessary (all of these operations are either
\texttt{insert} or \texttt{decrease\_key}, and both operations can be
done in constant time using the standard Fibonacci heap). That is, the
global queue contains at most $n$ values, potentially one from each
vertex queue.  In particular, once $c_v \geq k$, we disable the queue
$Q_v$, and it not longer participates in the global queue, or takes
updates to values in its queue.

\subsection{Analysis}

\begin{lemma}
    The algorithm running time is $O\pth{ k \pth{ n \log n + m}}$.
\end{lemma}

\begin{proof}
    Every vertex distance is going to be set by at most $k$ of these
    ``parallel'' executions of Dijkstra. Whenever it happens, the
    algorithm performs a relax operation on all the adjacent edges. It
    follows, that an edge would be relaxed $O( k)$ times. The cost of
    a relax operation is $O(1)$, and thus the overall cost of these is
    $O( k m)$. Setting one of the $k$ shortest distance values of a
    vertex takes $O(\log n)$ time, as it involves
    \texttt{extract\_min} from the global heap. As such, each vertex
    would require overall $O( k \log n)$ time for its
    \texttt{extract\_min} operations from the global heap and the
    local heap of this vertex.  The bound on the running time now
    follows readily.
\end{proof}

\subsection{The result}

\begin{theorem}%
    \thmlab{main}%
    Given a directed graph $\Graph = (\Vertices,\Edges)$ with $n$
    vertices and $m$ edges, with positive weights on the edges, and a
    parameter $k >0$, one can compute, in
    $O\pth{ k \pth{ n \log n + m}}$ time, for every vertex
    $v \in \Vertices$, its $k$ nearest-neighbors in $\Graph$.

    The algorithm uses hashing (i.e., randomization), a deterministic
    version of the algorithm runs in time
    $O\pth{ k n \log n + k m \log k }$.
\end{theorem}

\begin{proof}
    For the deterministic running time, replace the hash table
    maintained by each vertex $v$ (which stores all the
    nearest-neighbors to $v$ discovered so far) by a balanced binary
    tree.  The key observation is that such a tree, for a node $v$,
    needs to maintain only the $k$ smallest candidate distances
    offered to this vertex. In particular, if a new candidate distance
    is larger than all these $k$ values, we can immediately reject
    it. Similarly, after insertion we reject any value larger than $k$
    smallest values in this data-structure.

    As such, each binary tree stores at most $O(k)$ elements, so every
    basic operation takes $O( \log k)$ time. The algorithm performs
    $O( k(n+m))$ operations on these lookup data structures, which
    implies the claimed running time.
\end{proof}

\begin{remark}
    Consider the settings of \thmref{main}, but in addition there is a
    set of terminals $T \subseteq \Vertices$. One can compute for each
    vertex of $\Graph$ the $k$ closest \emph{terminals} to it in
    $\Graph$ using the algorithm of \thmref{main}. The only
    modification being that we start the Dijkstra process only from
    the vertices of $T$. The running time of the modified algorithm is
    the same.
\end{remark}


%


\paragraph*{Acknowledgments.}

The author thanks Chandra Chekuri and David Eppstein for useful
comments on this note. In addition, Vivek Madan came up independently
with a similar solution to the one described in this note.

 
 \providecommand{\CNFX}[1]{ {\em{\textrm{(#1)}}}}
  \providecommand{\tildegen}{{\protect\raisebox{-0.1cm}{\symbol{'176}\hspace{-0.03cm}}}}
  \providecommand{\SarielWWWPapersAddr}{http://sarielhp.org/p/}
  \providecommand{\SarielWWWPapers}{http://sarielhp.org/p/}
  \providecommand{\urlSarielPaper}[1]{\href{\SarielWWWPapersAddr/#1}{\SarielWWWPapers{}/#1}}
  \providecommand{\Badoiu}{B\u{a}doiu}
  \providecommand{\Barany}{B{\'a}r{\'a}ny}
  \providecommand{\Bronimman}{Br{\"o}nnimann}  \providecommand{\Erdos}{Erd{\H
  o}s}  \providecommand{\Gartner}{G{\"a}rtner}
  \providecommand{\Matousek}{Matou{\v s}ek}
  \providecommand{\Merigot}{M{\'{}e}rigot}
  \providecommand{\CNFSoCG}{\CNFX{SoCG}}
  \providecommand{\CNFCCCG}{\CNFX{CCCG}}
  \providecommand{\CNFFOCS}{\CNFX{FOCS}}
  \providecommand{\CNFSODA}{\CNFX{SODA}}
  \providecommand{\CNFSTOC}{\CNFX{STOC}}
  \providecommand{\CNFBROADNETS}{\CNFX{BROADNETS}}
  \providecommand{\CNFESA}{\CNFX{ESA}}
  \providecommand{\CNFFSTTCS}{\CNFX{FSTTCS}}
  \providecommand{\CNFIJCAI}{\CNFX{IJCAI}}
  \providecommand{\CNFINFOCOM}{\CNFX{INFOCOM}}
  \providecommand{\CNFIPCO}{\CNFX{IPCO}}
  \providecommand{\CNFISAAC}{\CNFX{ISAAC}}
  \providecommand{\CNFLICS}{\CNFX{LICS}}
  \providecommand{\CNFPODS}{\CNFX{PODS}}
  \providecommand{\CNFSWAT}{\CNFX{SWAT}}
  \providecommand{\CNFWADS}{\CNFX{WADS}}


\end{document}